\documentclass{article}
\usepackage[cmex10]{amsmath}
\usepackage{amsthm,amsfonts,paralist,mathtools}

\newcommand{\ket}[1]{|{#1}\rangle}
\newcommand{\bra}[1]{\langle{#1}|}
\newcommand{\braket}[2]{\langle{#1}\vert{#2}\rangle}
\newcommand{\nix}[1]{}

\DeclareMathOperator{\wt}{wt}

\newtheorem{proposition}{Proposition}
\newtheorem{theorem}[proposition]{Theorem}

\begin{document}

\title{Nonbinary Error-Detecting Hybrid Codes}
\author{Andrew Nemec and Andreas Klappenecker}
\date{}
\nocite{*}
\maketitle

\begin{abstract}
Hybrid codes simultaneously encode both quantum and classical information, allowing for the transmission of both across a quantum channel. We construct a family of nonbinary error-detecting hybrid stabilizer codes that can detect one error while also encoding a single classical bit over the residue class rings $\mathbb{Z}_{q}$ inspired by constructions of nonbinary non-additive codes.
\end{abstract}

\section{Introduction}

Hybrid codes allow for the simultaneous transmission of both quantum and classical information across a quantum channel. While it has long been known that simultaneous transmission can provide an advantage over the time-sharing of the channel for certain small error rates, see \cite{Devetak2005}, most of the early work on the topic focused on information-theoretic results, see \cite{Hsieh2010a, Hsieh2010b, Yard2005}, while the problem of constructing finite-length hybrid codes remained largely overlooked.

The first examples of hybrid codes were given by Kremsky, Hsieh, and Brun \cite{Kremsky2008}, who introduced them as a generalization of entanglement-assisted stabilizer codes. Later, Grassl, Lu, and Zeng \cite{Grassl2017} gave multiple examples of small hybrid codes constructed using an approach inspired by the construction of nonadditive codeword stabilized quantum codes. Remarkably, these codes provide an advantage over optimal quantum codes regardless of the error rate. Recently, several families of hybrid codes have been constructed including several families constructed by the authors \cite{Nemec2019} for the Pauli channel using stabilzer pasting and a family constructed by Li, Lyles, and Poon \cite{Li2019} for fully correlated quantum channels. An operator-theoretic approach to hybrid codes has also been put forward in \cite{Beny2007a, Beny2007b, Majidy2018}.

In \cite{Nemec2019}, the authors constructed several families of binary hybrid codes with good parameters, including a family of $\left[\!\left[n,n-3\!:\!1,2\right]\!\right]_{2}$ error-detecting codes where $n$ is odd. In this paper we provide a generalization of this family to hybrid stabilizer codes over $\mathbb{Z}_{q}$, inspired by the non-additive nonbinary quantum codes constructed from qudit graph states by Hu et al. \cite{Hu2008} and Looi et al. \cite{Looi2008}, as well as the family of single error-detecting codes given by Smolin, Smith, and Wehner \cite{Smolin2007}.

\subsection{Nonbinary Quantum Codes}
A quantum code is a subspace of a Hilbert space that allows for the recovery of encoded quantum information even in the presence of arbitrary errors on a certain number of physical qudits. A quantum code has parameters $\left(\!\left(n,K,d\right)\!\right)_{q}$ if and only if it can encode a superposition of $K$ orthogonal quantum states into the Hilbert space $\left(\mathbb{C}^{q}\right)^{\otimes n}\cong\mathbb{C}^{q^{n}}$, while protecting the quantum information against all errors ocurring on less than $d$ physical qubits.

Most generalizations of quantum codes from the binary alphabets to the case where $q>2$ are constructed over the finite fields $\mathbb{F}_{q}$, where $q$ is a prime power, see \cite{Ashikhmin2001, Ketkar2006, Rains1999b}. In this paper, we instead follow \cite{Hu2008, Looi2008, Smolin2007} and construct codes over $\mathbb{Z}_{q}$ for reasons that will become apparent in Section \ref{codeconstrsec}. Let $a,b\in\mathbb{Z}_{q}$. We define the unitary operators $X\!\left(a\right)$ and $Z\!\left(b\right)$ on $\mathbb{C}^{q}$ as $$X\!\left(a\right)\ket{x}=\ket{x+a}\text{ and } Z\!\left(b\right)\ket{x}=\omega^{bx}\ket{x},$$ where $\omega=e^{2\pi i/q}$. The operators $X\!\left(a\right)$ and $Z\!\left(b\right)$may be viewed as a generalization of the Pauli-$X$ bit-flip error and the Pauli-$Z$ phase error respectively. The set $\mathcal{E}=\left\{X\!\left(a\right)Z\!\left(b\right)\mid a,b\in\mathbb{Z}_{q}\right\}$ forms a nice error basis on $\mathbb{C}^{q}$, see \cite{Klappenecker2002, Klappenecker2003, Knill1996}, meaning any error on a single qudit may be written as a linear combination of elements from $\mathcal{E}$. Additionally, any error on $\mathbb{C}^{q^{n}}$ may be written as a linear combination of errors from $\mathcal{E}_{n}=\mathcal{E}^{\otimes n}=\left\{E_{1}\otimes E_{2}\otimes\cdots\otimes E_{n}\mid E_{k}\in\mathcal{E}, 1\leq k\leq n\right\}$. By correcting errors from $\mathcal{E}_{n}$ we are able to deal with arbitrary errors on the $n$ qudits that are linear combinations of those errors. The weight $\wt\!\left(E\right)$ of an error $E\in\mathcal{E}_{n}$ is the number of non-identity tensor components it contains.

A quantum code $\mathcal{C}$ has the ability to detect an error $E\in\mathcal{E}_{n}$ if it either reports than an error occured or reports no error and returns a projection of the message back onto $\mathcal{C}$. Formally, the Knill-Laflamme conditions tell us that an error $E$ is detectable by a quantum code $\mathcal{C}$ if and only if $PEP=\lambda_{E}P$ for some scalar $\lambda_{E}$, where $P$ is the orthogonal projector onto $\mathcal{C}$, see \cite{Knill1997}.

Stabilizer codes are perhaps the most important class of quantum codes, and are analogous to the linear and additive codes in classical coding theory (hence they are also refered to as additive codes). Stabilizer codes are completely determined by their stabilizer group $\mathcal{S}$, an abelian subgroup of $\mathcal{E}_{n}$, and the code is defined as the subspace spanned by all joint eigenvectors of $\mathcal{S}$ with eigenvalue 1. Since this subspace will always have dimension $K=q^{k}$, we say the code has parameters $\left[\!\left[n,k,d\right]\!\right]_{q}$ to denote it as a stabilizer code.

\subsection{Hybrid Codes}

In addition to transmitting quantum information, we now want to simultaneously encode a classical message in with the encoded quantum state. A hybrid code has parameters $\left(\!\left(n,K\!:\!M,d\right)\!\right)_{q}$ if and only if it can simultaneously encode a superposition of $K$ orthogonal quantum states as well as one of $M$ different classical states into $\left(\mathbb{C}^{q}\right)^{\otimes n}\cong\mathbb{C}^{q^{n}}$, a Hilbert space of dimension $q^{n}$, while protecting both the quantum and classical information against all errors of weight less than $d$. 

An $\left(\!\left(n,K\!:\!M,d\right)\!\right)_{q}$ hybrid code $\mathcal{C}$ can be described by a collection of $M$ orthogonal quantum codes $\mathcal{C}_{m}$ of dimension $K$, each indexed by a classical message $m\in\left[M\right]=\left\{0,1,\dots, M-1\right\}$. To transmit a quantum state $\varphi$ and a classical message $m$, we encode $\varphi$ into the quantum code $\mathcal{C}_{m}$. The Knill-Laflamme conditions for quantum codes can be generalized to hybrid codes, allowing us to characterize detectable errors: an error $E$ is detectable by the hybrid code $\mathcal{C}$ if and only if \begin{equation}
P_{b} E P_{a} = \begin{cases}
\lambda_{E,a} P_a & \text{if $a=b$}, \\
0 & \text{if $a\neq b$}
\end{cases},
\end{equation}
for some scalar $\lambda_{E,a}$ depending on both the error $E$ and the classical message $a$, where $P_{a}$ is the orthogonal projector onto the quantum code $\mathcal{C}_{a}$. Equivalently, if $\left\{\ket{c_{i}^{\left(m\right)}}\right\}$ are the codewords of the inner code $\mathcal{C}_{m}$, we have that $E$ is detectable by $\mathcal{C}$ if and only if \begin{equation} \bra{c_{j}^{\left(b\right)}}E\ket{c_{i}^{\left(a\right)}}=\lambda_{E,a}\delta_{i,j}\delta_{a,b}.\end{equation}

If both the inner codes and outer code happen to be stabilizer codes, we say the code is a hybrid stabilizer code with parameters $\left[\!\left[n,k\!:\!m,d\right]\!\right]_{q}$. In this case, the codes have some additional structure, so each inner code can be viewed as a translation from a seed code $\mathcal{C}_{0}$ by an operator $t_{m}\in\mathcal{E}_{n}\setminus Z\!\left(\mathcal{S}_{0}\right)$, where $Z\!\left(\mathcal{S}_{0}\right)$ is the centralizer in $\mathcal{E}_{n}$ of the stabilizer $\mathcal{S}_{0}$ of the seed code $\mathcal{C}_{0}$, so that $\mathcal{C}_{m}=t_{m}\mathcal{C}_{0}$.

There are multiple simple constructions of hybrid codes using quantum codes described by Grassl et al. \cite{Grassl2017}:

\begin{proposition}\label{trivial}
Hybrid codes can be constructed using the following ``trivial" constructions:
\begin{enumerate}
\item Given an $\left(\!\left(n,KM,d\right)\!\right)_{q}$ quantum code of composite dimension $KM$, there exisits a hybrid code with parameters $\left(\!\left(n,K\!:\!M,d\right)\!\right)_{q}$.
\item Given an $\left[\!\left[n,k\!:\!m,d\right]\!\right]_{q}$ hybrid code with $k>0$, there exists a hybrid code with parameters $\left[\!\left[n,k-1\!:\!m+1,d\right]\!\right]_{q}$.
\item Given an $\left[\!\left[n_{1},k_{1},d\right]\!\right]_{q}$ quantum code and an $\left[n_{2},m_{2},d\right]_{q}$ classical code, there exists a hybrid code with parameters $\left[\!\left[n_{1}+n_{2},k_{1}\!:\!m_{2},d\right]\!\right]_{q}$.
\end{enumerate}
\end{proposition}

In each of these cases the sender is effectively substituting classical information for quantum information, which depending on the context may be considered wasteful. In \cite{Grassl2017}, Grassl et al. showed it was possible to construct \emph{genuine} hybrid codes that provide an advantage over these simple codes, and provided examples of such codes found using an exhaustive search of small parameters. In \cite{Nemec2019} the authors constructed several infinite families of genuine hybrid codes, including a family of binary single error-detecting codes which we generalize to the nonbinary case in the next section.

\section{Family of Hybrid Codes over $\mathbb{Z}_{q}$}\label{codeconstrsec}

The first good non-additive quantum code (that is a quantum code that is not a stabilizer code) was the $\left(\!\left(5,6,2\right)\!\right)_{2}$ code given by Rains et al. \cite{Rains1997}. This code outperforms the optimal $\left[\!\left[5,2,2\right]\!\right]_{2}$ stabilizer code, and was further generalized by Rains \cite{Rains1999a} into a family of odd-length non-additive codes that outperform optimal stabilizer codes. However, for an odd-length $\left(\!\left(n, K, 2\right)\!\right)$ quantum code we have the following bound: \begin{equation}\label{rainsbound} K\leq2^{n-2}\left(1-\frac{1}{n-1}\right),\end{equation} and many families of codes that approach this bound have been constructed. In \cite{Nemec2019}, the authors gave a construction for a familiy of hybrid stabilizer codes with parameters $\left[\!\left[n,n-3\!:\!1,2\right]\!\right]_{2}$ that beat this bound.

Nonbinary quantum codes with similar parameters were hinted at by Rains in \cite{Rains1999a}, and first given by Smolin et al. \cite{Smolin2007} as a generalization of their family of non-additive binary codes. Soon after, further families were constructed by Hu et al. \cite{Hu2008} and Looi et al. \cite{Looi2008} using qudit graph states. All of these families are codes over integer rings rather than finite fields, and our construction of nonbinary  hybrid stabilizer codes will follow in their footsteps. The reason we choose to construct codes over $\mathbb{Z}_{q}$ rather than $\mathbb{F}_{q}$ is due to the following result of Grassl and R{\"o}tteler:

\begin{theorem}[{\cite[Theorem 12]{Grassl2015}}]
Let $q>1$ be an arbitrary integer, not necessarily a prime power. Quantum MDS codes $\mathcal{C}=\left[\!\left[n,n-2,2\right]\!\right]_{q}$ exist for all even length $n$, and for all length $n\geq 2$ when the dimension $q$ of the quantum systems is an odd integer or is divisible by 4.
\end{theorem}

While the construction below will certainly produce a hybrid stabilizer code when $q\not\equiv2\mod 4$, it will not be a genuine hybrid code, as the previous theorem implies that there will be an $\left[\!\left[n,n-2,2\right]\!\right]_{q}$ stabilizer code that can be transformed into a hybrid code using the first construction in Proposition \ref{trivial}. When $q=2$, Equation \ref{rainsbound} tells us that there can be no $\left[\!\left[n,n-2,2\right]\!\right]_{2}$ quantum code, implying that the family given in \cite{Nemec2019} is indeed genuine. To the best of our knowledge there are no known $\left[\!\left[n,n-2,2\right]\!\right]_{q}$ codes when $q=4r+2$, which is why the codes using the construction below may in fact be genuine. However, since $\mathbb{F}_{4r+2}$ does not exist except when $r=0$, we instead construct our codes over $\mathbb{Z}_{q}$.

\begin{proposition}
Let $n$ be odd. Then there exists an $\left[\!\left[n,n-3\!:\!1,2\right]\!\right]_{\mathbb{Z}_{q}}$ hybrid code.
\end{proposition}
\begin{proof}
Let $a,b\in\mathbb{Z}_{q}^{n}$, $m\in\mathbb{Z}_{q}$, and $\omega$ a primitive $q$-th root of unity. Define the following states:
\begin{equation*}
\ket{\phi_{a,b}}=\frac{1}{q^{n}}\sum\limits_{c\in\mathbb{Z}_{q}^{2n}}\omega^{\sum_{i=1}^{n}\left(c_{2i-1}-a_{i}\right)\left(c_{2i}-b_{i}\right)}\ket{c}
\end{equation*}
\begin{equation*}
\ket{\psi_{m}}=\frac{1}{\sqrt{q}}\sum\limits_{c\in\mathbb{Z}_{q}}\omega^{mc}\ket{c}
\end{equation*}
Define the inner code $\mathcal{C}_{m}$ as follows:
\begin{equation*}
\mathcal{C}_{m}=\left\langle\ket{\phi_{a,b}}\otimes\ket{\psi_{m}}\middle|a,b\in\mathbb{Z}_{q}^{n},m\in\mathbb{Z}_{q}\sum\limits_{i=1}^{n}a_{i}=0, \sum\limits_{i=1}^{n}b_{i}=m\right\rangle
\end{equation*}

The state $\ket{\phi_{a,b}}$ is the tensor product of two-qubit states of the form
\begin{equation*}
\ket{\phi_{a_{i},b_{i}}}=\frac{1}{q}\sum\limits_{c\in\mathbb{Z}_{q}^{2}}\omega^{\left(c_{1}-a_{i}\right)\left(c_{2}-b_{i}\right)}\ket{c}.
\end{equation*}
For two of these states $\ket{\phi_{a_{i},b_{i}}},\ket{\phi_{a'_{i},b'_{i}}}$ we have
\begin{align*}
\braket{\phi_{a_{i},b_{i}}}{\phi_{a'_{i},b'_{i}}} & = \frac{1}{q^{2}}\sum\limits_{c\in\mathbb{Z}_{q}^{2}}\omega^{\left(c_{1}-a'_{i}\right)\left(c_{2}-b'_{i}\right)-\left(c_{1}-a_{i}\right)\left(c_{2}-b_{i}\right)} \\
& = \frac{\omega^{a'_{i}b'_{i}-a_{i}b_{i}}}{q^{2}}\sum\limits_{c\in\mathbb{Z}_{q}^{2}}\omega^{c_{1}\left(b'_{i}-b_{i}\right)+c_{2}\left(a'_{i}-a_{i}\right)} \\
& = \frac{\omega^{a'_{i}b'_{i}-a_{i}b_{i}}}{q^{2}}\left(\sum\limits_{c_{1}\in\mathbb{Z}_{q}}\omega^{c_{1}\left(b'_{i}-b_{i}\right)}\right)\left(\sum\limits_{c_{2}\in\mathbb{Z}_{q}}\omega^{c_{2}\left(a'_{i}-a_{i}\right)}\right) \\
& = \begin{cases} 1 & \text{if } a_{i}=a'_{i} \text{ and } b_{i}=b'_{i} \\ 0 & \text{otherwise} \end{cases}.
\end{align*}
Therefore for the full states $\ket{\phi_{a,b}},\ket{\phi_{a',b'}}$ we have the same:
\begin{equation*}
\braket{\phi_{a,b}}{\phi_{a',b'}}=\begin{cases} 1 & \text{if } a=a' \text{ and } b=b' \\ 0 & \text{otherwise} \end{cases}.
\end{equation*}
Similarly, for $\ket{\psi_{m}},\ket{\psi_{m'}}$ we have
\begin{equation*}
\braket{\psi_{m}}{\psi_{m'}}=\begin{cases} 1 & \text{if } m=m' \\ 0 & \text{otherwise} \end{cases}.
\end{equation*}
Thus all of the codewords are orthogonal to one another.

Consider two codewords $\ket{\phi_{a,b}}\otimes\ket{\psi_{m}},\ket{\phi_{a',b'}}\otimes\ket{\psi_{m'}}$. Suppose that a Pauli-$X\!\left(u\right)$ error occurs on the first $n-1$ qudits. Without loss of generality, we can assume that the error occurred on either the first or second qudit. If $m\neq m'$, $a_{i}\neq a'_{i}$, or $b_{i}\neq b'_{i}$ for $1<i\leq n$, then $$\left(\bra{\phi_{a,b}}\otimes\bra{\psi_{m}}\right)X\!\left(u\right)\left(\ket{\phi_{a',b'}}\otimes\ket{\psi_{m'}}\right)=0$$ by the orthogonality relations above. Therefore we can restrict our attention to the case where $m=m'$, $a_{i}=a'_{i}$, and $b_{i}=b'_{i}$ for $1<i\leq n$). We note that these restrictions along with the requirement that the $a_{i}$ and $a'_{i}$ sum to 0 and $b_{i}$ and $b'_{i}$ sum to $m$ and $m'$ respectively completely determine the values of $a_{1}$ and $b_{1}$ and in particular we must have $a_{1}=a'_{1}$ and $b_{1}=b'_{1}$. If the error occurred on the first qudit, we have 
\begin{align*}
\bra{\phi_{a_{1},b_{1}}}X\!\left(u\right)\ket{\phi_{a_{1},b_{1}}} & = \frac{1}{q^{2}}\left(\sum\limits_{c\in\mathbb{Z}_{q}^{2}}\omega^{-\left(c_{1}-a_{1}\right)\left(c_{2}-b_{1}\right)}\bra{c_{1}c_{2}}\right)\left(\sum\limits_{c\in\mathbb{Z}_{q}^{2}}\omega^{\left(c_{1}-a_{1}\right)\left(c_{2}-b_{1}\right)}\ket{\left(c_{1}+u\right)c_{2}}\right) \\
& = \frac{1}{q^{2}}\left(\sum\limits_{c\in\mathbb{Z}_{q}^{2}}\omega^{-\left(c_{1}-a_{1}\right)\left(c_{2}-b_{1}\right)}\bra{c_{1}c_{2}}\right)\left(\sum\limits_{c\in\mathbb{Z}_{q}^{2}}\omega^{\left(c_{1}-a_{1}-u\right)\left(c_{2}-b_{1}\right)}\ket{c_{1}c_{2}}\right) \\
& = \frac{1}{q^{2}}\sum\limits_{c_{2}\in\mathbb{Z}_{q}^{2}}\omega^{u\left(b_{1}-c_{2}\right)} \\
& = \begin{cases} 1 & \text{if } u=0 \\ 0 & \text{otherwise} \end{cases}.
\end{align*}
A similar argument holds if the error occurs on the second qudit, thus the code can detect any single Pauli-$X\!\left(u\right)$ error that occurs on the first $n-1$ qudits.

Now suppose that a Pauli-$Z\!\left(v\right)$ error occurs on the first $n-1$ qudits. As above, we restrict our attention to the case where $a=a'$, $b=b'$, $m=m'$, and the error occurs on one of the first two qudits. If the error occurs on the first qudit we have
\begin{align*}
\bra{\phi_{a_{1},b_{1}}}Z\!\left(v\right)\ket{\phi_{a_{1},b_{1}}} & = \frac{1}{q^{2}}\sum\limits_{c\in\mathbb{Z}_{q}^{2}}\omega^{\left(c_{1}-a_{1}\right)\left(c_{2}-b_{1}\right)-\left(c_{1}-a_{1}\right)\left(c_{2}-b_{1}\right)+vc_{1}} \\
& = \frac{1}{q}\sum\limits_{c_{1}\in\mathbb{Z}_{q}}\omega^{vc_{1}} \\
& = \begin{cases} 1 & \text{if } v=0 \\ 0 & \text{otherwise} \end{cases}.
\end{align*}
The same argument holds if the error occurs on the second qudit, thus the code can detect any single Pauli-$Z\!\left(v\right)$ error that occurs on the first $n-1$ qudits.

Now suppose that a Pauli error $E$ occurs on the last qudit. If $a\neq a'$, $b\neq b'$, or $m\neq m'$, then the orthogonality of the first $n-1$ qudits gives us $$\left(\bra{\phi_{a,b}}\otimes\bra{\psi_{m}}\right)E\left(\ket{\phi_{a',b'}}\otimes\ket{\psi_{m'}}\right)=0,$$ so again we only need to examine the case where the two codewords are the same.

If we have a Pauli-$X\!\left(u\right)$ error on the last qudit we have
\begin{align*}
\bra{\psi_{m}}X\!\left(u\right)\ket{\psi_{m}} & = \frac{1}{q}\left(\sum\limits_{c\in\mathbb{Z}_{q}}\omega^{-mc}\bra{c}\right)\left(\sum\limits_{c\in\mathbb{Z}_{q}}\omega^{mc}\ket{c+u}\right) \\
& = \frac{1}{q}\sum\limits_{c\in\mathbb{Z}_{q}}\omega^{-mu} \\
& = \omega^{-mu},
\end{align*}
meaning that the error is degenerate. Note that since the value depends on the classical information $m$, each inner code can detect the error but the outer code (as a quantum code) cannot.

If a Pauli-$Z\!\left(v\right)$ error occurs on the last qudit we have
\begin{align*}
\bra{\psi_{m}}Z\!\left(v\right)\ket{\psi_{m}} & = \frac{1}{q}\left(\sum\limits_{c\in\mathbb{Z}_{q}}\omega^{-mc}\bra{c}\right)\left(\sum\limits_{c\in\mathbb{Z}_{q}}\omega^{mc+vc}\ket{c}\right) \\
& = \frac{1}{q}\sum\limits_{c\in\mathbb{Z}_{q}}\omega^{vc} \\
& = \begin{cases} 1 & \text{if } v=0 \\ 0 & \text{otherwise} \end{cases}.
\end{align*}

\end{proof}

We also mention in passing that this construction can be generalized further to codes over Frobenius rings by replacing the primitive root of unity by an irreducible additive character of the additive group of the ring \cite{Nadella2012}.

\section{Conclusion and Discussion}

Hybrid codes simultaneously transmit both quantum and classical information across quantum channels, and can provide an advantage over using quantum codes for simultaneous transmission. We have generalized a family of single error-detecting codes constructed in \cite{Nemec2019} from the binary case to the nonbinary case. While it is known that the construction gives genuine hybrid codes when $q=2$, the existence of quantum codes with the similar parameters when $q\equiv0,1,3\mod 4$ means the construction does not produce genuine hybrid codes in all cases. One open question is whether or not the codes given by the construction are always genuine when $q\equiv2 \mod 4$. As the code family here is the only construction of nonbinary hybrid codes, further investigation is needed. 

\bibliographystyle{IEEEtran}

\end{document}